\documentclass{ws-ijqi}

\usepackage{url}

\newcommand{\ket}[1]{|#1\rangle}
\newcommand{\bra}[1]{\langle#1|}
\newcommand{\proj}[1]{\ket{#1}\bra{#1}}
\newcommand{\braket}[2]{\langle #1 | #2 \rangle}
\newcommand{\wek}[1]{{\boldsymbol{#1}}}

\newcommand{\beq}{\begin{equation}}
\newcommand{\eeq}{\end{equation}}

\newcommand{\Tr}{{\rm Tr}}

\newcommand{\cB}{\mathcal{B}}

 \newcommand{\ketbra}[2]{|#1\rangle\!\langle #2|}

\newcommand{\nc}{\newcommand}
\nc{\ox}{\otimes}

\newcommand{\trace}{{\rm Tr}}

\newcommand{\set}[1]{{\left\{#1\right\}}}    

\newcommand{\abs}[1]{\left\lvert #1 \right\rvert}

\newcommand{\complex}{{\mathbb C}}
\newcommand{\reals}{{\mathbb R}}

\newcommand{\openone}{I}

\usepackage{color,graphicx}

\begin{document}


\catchline{}{}{}{}{}

\title{CHARACTERIZING QUANTUMNESS VIA ENTANGLEMENT CREATION}

\author{Sevag Gharibian}
\address{Institute for Quantum Computing and School of Computer Science,\\
 University of Waterloo, Waterloo ON N2L 3G1, Canada\\
sggharib@cs.uwaterloo.ca}
\author{Marco Piani}
\address{Institute for Quantum Computing and Department of Physics and Astronomy,\\
University of Waterloo, Waterloo ON N2L 3G1, Canada\\
mpiani@iqc.ca}
\author{Gerardo Adesso}
\address{School of Mathematical Sciences,\\
University of Nottingham, University Park, Nottingham NG7 2RD, U.~K.\\
gerardo.adesso@nottingham.ac.uk}
\author{John Calsamiglia}
\address{F\'{i}sica Te\`{o}rica: Informaci\'{o} i Fen\`{o}mens Qu\`{a}nitcs,\\
Universitat Aut\`{o}noma de Barcelona, 08193 Bellaterra, Spain\\
john.calsamiglia@uab.cat}
\author{Pawe{\l} Horodecki}
\address{Faculty of Applied Physics and Mathematics,\\
Technical University of Gda\'nsk, 80-952 Gda\'nsk, Poland\\
pawel@mif.pg.gda.pl}

%
%
%

\maketitle


\begin{abstract}
In [M. Piani et al., arXiv:1103.4032 (2011)]  an activation protocol was introduced which maps the general non-classical (multipartite) correlations between given systems into bipartite entanglement between the systems and local ancillae by means of a potentially highly entangling interaction. Here, we study how this activation protocol can be used to entangle the starting systems themselves via entanglement swapping through a measurement on the ancillae. Furthermore, we bound the relative entropy of quantumness (a naturally arising measure of non-classicality in the scheme of Piani et al. above) for a special class of separable states, the so-called classical-quantum states. In particular, we fully characterize the classical-quantum two-qubit states that are maximally non-classical.
\end{abstract}


\section{Introduction}
\label{sec:intro}

The non-classicality of correlations present in bi- and multi-partite quantum states is not due solely to the presence of entanglement.\cite{HHHH09} Namely, there exist quantum states which are unentangled, but nevertheless exhibit traits that have \emph{no} counterpart in the classical world. Such traits include (e.g.) no-local broadcasting\cite{pianietal2008nolocalbrodcast} and the locking of correlations.\cite{DHLST04,DG09} Furthermore, a notion of non-classicality weaker than entanglement --- or rather, more general than entanglement ---  is conjectured to play a role in the model of mixed-state quantum computation known as DQC1,\cite{PhysRevLett.81.5672} providing a believed exponential speed-up with respect to classical computation even in the absence or limited presence of entanglement.\cite{PhysRevLett.100.050502} This latter conjecture in particular
has led to significant efforts to characterize and quantify the non-classicality --- or \emph{quantumness} --- of correlations.\cite{PhysRevLett.88.017901,HV,PhysRevA.71.062307,PhysRevA.72.032317,PhysRevLett.104.080501,groismanquantumness,PhysRevA.77.052101,Luo2008,Bravyi2003,pianietal2008nolocalbrodcast,pianietal2009broadcastcopies,ferraro,ADA,PhysRevA.82.052342,streltsov2011}

In this latter context, an \emph{activation} protocol was considered in Ref.~\refcite{activation} which maps general non-classical (multipartite) correlations between input systems into bipartite entanglement between the systems and ancillae. This is accomplished by letting the ancillae and input systems interact via highly entangling gates, namely CNOTs, with the systems acting as controls (a formal description of the protocol is given in Section~\ref{sec:definitions}). One advantange of this mapping is that it allows us to apply the tools and concepts of entanglement theory to the study of the quantumness of correlations. Further, the activation protocol, when considered in an adversarial context where the control bases are chosen so as to create the minimal amount of system-ancilla entanglement, provides an operational interpretation of the relative entropy of quantumness as being the minimum distillable entanglement\cite{reviewplenio} \emph{necessarily} (i.e. in the worst case scenario) created between the input systems and the ancillae.

In this article, we present two main contributions towards a better understanding of the quantumness of correlations, both inspired by the activation protocol of Ref.~\refcite{activation}.

The first contribution is a non-trivial upper bound on the relative entropy of quantumness\cite{Bravyi2003,groismanquantumness,PhysRevA.77.052101,Luo2008,PhysRevLett.104.080501} for a special class of separable states, the so-called classical-quantum states. Further, we are able to fully characterize the classical-quantum two-qubit states which are maximally non-classical (with respect to the relative entropy of quantumness).

The second contribution is the study of an approach for entangling the input systems via the use of the ancillae that are introduced in the activation protocol.
The approach is as follows: We first let each system interact with an ancilla, and then we try to ``swap''\cite{swapping} the entanglement created between the input systems and ancillae back into entanglement among the input systems. Further, we assume a worst-case scenario in performing this mapping --- namely, we are interested in whether there exists a choice of control bases for the activation protocol for which no entanglement can be created between the input systems with this approach, even if we allow for post-selection in the manipulation of the ancillae after the entangling interaction.
With respect to this mapping, we find conditions under which entanglement can or cannot be swapped into the input system. For example, we find that there are non-classical states such that an adversarial choice of the control bases makes it impossible to swap  entanglement into the input systems, even if entanglement is \emph{necessarily} created between systems and ancillae in the first interaction step.

This paper is organized as follows. In Section~\ref{sec:definitions}, we state definitions and recall the activation protocol introduced in Ref.~\refcite{activation}. In Section~\ref{sec:bounds}, we provide bounds on non-classicality, as measured by the relative entropy of quantumness. In Section~\ref{sec:swapping}, we present several results and observations regarding entangling input systems via the activation protocol and entanglement swapping. In Section~\ref{sec:conclusions}, we conclude.

\section{Background and Definitions}
\label{sec:definitions}

We first define a \emph{classical} quantum state as follows.
\begin{definition}[Strictly Classically Correlated Quantum State]
\label{def:classical}
Given a set of n $d$-dimensional qudit systems, let ${\cB}_i$ denote an orthonormal basis in $\complex^d$ for the $i$th system consisting of vectors $\ket{\cB_i(k)}$ for $0\leq k\leq d-1$, and let ${\cB}$ denote an orthonormal basis $\{\ket{\cB(\wek{k})}=\ket{\cB_1(k_1)}\ket{\cB_2(k_2)}\cdots \ket{\cB_n(k_n)}\}$ for the entire space $(\complex^{d})^{\otimes n}$ formed by taking tensor products of all elements in bases $\set{{\cB}_i}_{i=1}^n$.
Then, an $n$-qudit state $\rho$ is \emph{strictly classically correlated}---or  simply \emph{classical}---if there exists such a basis $\cB$ with respect to which $\rho$ is diagonal. Such states correspond to the embedding of a multipartite classical probability distribution into the quantum formalism.
\end{definition}

We now outline the activation scheme of Ref.~\refcite{activation}, and follow with a formal description. Roughly, the scheme consists of letting the input systems under scrutiny interact with local ancillae, and then studying the entanglement generated between systems and ancillae. While we consider a potentially very entangling operation---specifically, given by CNOTs, with the input systems acting as control registers---we assume a worst-case-scenario perspective where we are interested in a choice of control bases for the CNOTs which are worst for entanglement generation. In other words, one can view the protocol as a game in which one tries to adversarially \emph{minimize} the entanglement produced under the action of the CNOTs by cleverly choosing the control basis of the CNOTs. Intuitively, this leads to entanglement generation for non-classical states because a CNOT gate entangles generic states except those that belong to a particular basis set --- in particular, it is impossible to write non-classical states as a convex combination of states from this latter set.

To be more precise, 
let $\wek{A}$ denote a joint register $A_1,\ldots,A_n$ of $n$ qudit systems (i.e. density operators representing states in $\wek{A}$ act on $(\complex^{d})^{\otimes n}$), henceforth called the ``system'', and let $\wek{A'}$ denote a joint register $A'_1,\ldots,A'_n$ of $n$ ancilla qudit registers, henceforth called the ``ancilla'' (see Figure~\ref{fig:activation}).
\begin{figure}[pb]
\centerline{\psfig{file=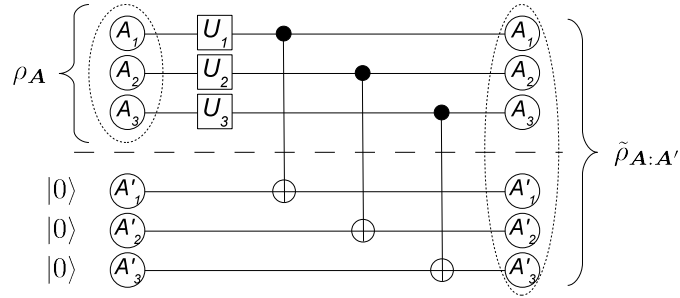,width=9cm}}
\vspace*{8pt}
\caption{Scheme of the activation protocol for $n=3$.}
\label{fig:activation}
\end{figure}
The initial state of the total $2n$ qudits is a tensor product between system and ancilla, namely $\rho_{\wek{A}:\wek{A'}} = \rho_{\wek{A}} \otimes \ket{0}\bra{0}^{\otimes n}_{\wek{A'}}$.
For a given input $\rho_\wek{A}$, we first consider for each $i$ an adversarial application of a local unitary $U_i$ to each $A_i$ (i.e. this chooses the control basis for system $i$), and follow by applying one CNOT gate on each subsystem $A_i$ (control qudit) and the corresponding ancillary party $A'_i$ (target qudit). The action of the CNOT on the computational basis states $\ket{j}\ket{j^\prime}$ of $\complex^d \otimes \complex^d$ is defined as
$\ket{j}\ket{j^\prime}\mapsto\ket{j}\ket{j^\prime\oplus j}$,
with $\oplus$ denoting addition modulo $d$. Formally, the final state of system plus ancilla at the end of the protocol is
\begin{equation}\label{eq:finalstate}
\tilde{\rho}_{\wek{A}:\wek{A'}} = V (\rho_{\wek{A}} \otimes \ket{0}\bra{0}^{\otimes n}_{\wek{A'}}) V^\dagger\,,
\end{equation}
with $V = {{CNOT}}_{\wek{A}:\wek{A'}} \cdot (U_{\wek{A}} \otimes \openone_{\wek{A'}})$ and $U_{\wek{A}} = \otimes_{i=1}^n U_i$. The question now is: What is the minimum amount of entanglement generated in the $\wek{A}:\wek{A'}$ split over all choices of local control bases $U_i$?

We first note that the state $\tilde{\rho}_{\wek{A}:\wek{A'}}$ in the mapping \eqref{eq:finalstate} can be rewritten as
\beq
\label{eq:maxcorr}
\tilde{\rho}_{\wek{A}:\wek{A'}}=\sum_{\wek{i}\wek{j}}\rho_{\wek{i}\wek{j}}^\wek{\cB}\ket{\wek{i}}\bra{\wek{j}}_{\wek{A}}\otimes \ket{\wek{i}}\bra{\wek{j}}_{\wek{A'}},
\eeq
where
\beq
\label{eq:matrixelement}
\rho_{\wek{i}\wek{j}}^\wek{\cB}=\bra{\wek{\cB}(\wek{i})}\rho_\wek{A} \ket{\wek{\cB}(\wek{j})},
\eeq
and  $\ket{\wek{\cB}(\wek{i})}_{\wek{A}}=U_\wek{A}^\dagger \ket{\wek{i}}$. In other words, $\tilde{\rho}_{\wek{A}:\wek{A'}}$ is of the maximally correlated\cite{R01} form in the $\wek{A}:\wek{A'}$ cut.

In Ref.~\refcite{activation}, it was proven that an input state $\rho_{\wek{A}}$ is classical if and only if there exists a choice of $U_\wek{A}$ such that  $\tilde{\rho}_{\wek{A}:\wek{A'}}$ is not entangled in the $\wek{A}:\wek{A'}$ bipartite cut. In particular one can choose to quantify the bipartite entanglement generated across the $\wek{A}:\wek{A'}$ cut by the relative entropy of entanglement,\cite{PhysRevLett.78.2275,PhysRevA.57.1619}
\beq
E_R(\rho_{C:D})=\min_{\text{separable }\sigma_{C:D}}S(\rho_{C:D}\|\sigma_{C:D}),
\eeq
where we define the relative entropy as $S(\rho\|\sigma):=\trace(\rho \log_2 \rho - \rho \log_2 \sigma)$, and by \emph{separable} $\sigma_{C:D}$ we mean states admitting a separable decomposition $\sigma_{CD}=\sum_i p_i \sigma^i_C \otimes \sigma^i_D$.
We remark that since the output of the mapping in \eqref{eq:finalstate} is maximally correlated, we have that $E_R(\tilde{\rho}_{\wek{A}:\wek{A'}})=E_D(\tilde{\rho}_{\wek{A}:\wek{A'}})$,\cite{Hiroshima2004} where $E_D$ denotes the distillable entanglement.\cite{reviewplenio} We thus obtain the following measure of non-classicality (see Ref.~\refcite{activation} for a proof of the various equivalences):
 \begin{equation}
 \label{eq:equivalences}
\begin{aligned}
 Q({\rho}_{{\wek{A}}}):&=\min_{U_{{\wek{A}}}}E_{D}(\tilde{\rho}_{\wek{A}:\wek{A'}})\\
 &=\min_{U_{{\wek{A}}}}E_{R}(\tilde{\rho}_{\wek{A}:\wek{A'}})\\
 &=\min_{\textrm{classical}~\sigma_\wek{A}}S(\rho_\wek{A}\|\sigma_\wek{A})\\
 &=\min_{{\cB}}\Big(S(\rho_{\wek{A}}^{\wek{\cB}})-S(\rho_{\wek{A}})\Big),
\end{aligned}
 \end{equation}
which is also known as the \emph{relative entropy of quantumness}.\cite{Bravyi2003,groismanquantumness,PhysRevA.77.052101,Luo2008,PhysRevLett.104.080501}
Here, the minimization in the first two lines is over local unitaries, in the third line is over all classical states $\sigma_\wek{A}$ and in the last line is over all choices of local orthonormal bases. We denote by $S(\sigma)$ the von Neumann entropy $S(\sigma):=-\Tr(\sigma\log_2\sigma)$, and by $\rho_{\wek{A}}^{\wek{{\cB}}}$ the state resulting from local projective measurements in the local bases $\wek{\cB}$ on $\rho_{\wek{A}}$, i.e. $\rho_{\wek{A}}^{\wek{\cB}}=\sum_{\wek{i}}\proj{\wek{\cB}(\wek{i})}\rho_\wek{A} \proj{\wek{\cB}(\wek{i})}$.

\section{Upper Bounds for Classical-Quantum Separable States}\label{scn:upperBoundsSeparable}
\label{sec:bounds}

In Ref.~\refcite{activation}, it was shown that for a bipartite state $\rho_{AB}$ (where in the bipartite case we adopt the notational convention that $A_1=A$ and $A_2=B$) the quantity $Q(\rho_{AB})$ can achieve its maximum value only for entangled states. It was also found that for asymptotically increasing local dimension $d$, separable states can be almost as non-classical as pure entangled states. Here, we consider the complementary problem of the maximum value attainable for $Q(\rho_{AB})$ by \emph{separable} states of \emph{fixed} dimensions. What we are able to obtain is a simple upper bound on $Q(\rho_{AB})$ for classical-quantum (CQ) states that holds for arbitrary local dimensions. Here, CQ states are separable states which are block diagonal for some choice of basis for $A$ (see Lemma \ref{l:cqbound} below). We then completely characterize the set of maximally non-classical two-qubit CQ states with respect to the relative entropy of quantumness $Q(\rho_{AB})$, and show that such states achieve $Q(\rho_{AB})=1/2$.

We first derive the claimed upper bound, which holds even when the local dimensions of $A$ and $B$ differ. We remark that, although for simplicity the case of equal dimensions is considered (both here and in Ref.~\refcite{activation}), a generalization to different dimensions for the input systems (matched by corresponding dimensions on the ancilla side) is straightforward for most of the results of~Ref.~\refcite{activation}. In particular, the relations \eqref{eq:equivalences} still hold true.

\begin{lemma}\label{l:cqbound}
For any CQ state $\rho_{AB} = \sum_{i=1}^{d_A} p_i \ketbra{i}{i}\otimes\rho_i$, where $\set{\ket{i}}_{i=1}^{d_A}$ is an orthonormal basis and $d_A$ and $d_B$ denote the local dimensions of systems $A$ and $B$, respectively, one has
\begin{equation}
    Q_{E_\textup{D}} (\rho_{AB})\leq \left(1-\frac{1}{d_A}\right)\log_2 d_B.
\end{equation}

\end{lemma}

\begin{proof}
We have
\begin{eqnarray}
    Q (\rho_{AB}) &=& \min_{\wek{\mathcal{B}}} S(\rho_{AB}^{\wek{\mathcal{B}}})-S(\rho_{AB})\nonumber\\
        &=& \min_{\mathcal{B}_B} S\left(\sum_{i=1}^{d_A} p_i \ketbra{i}{i}\otimes\left(\sum_{{j=1}}^{d_B}\proj{{\cB}_B({j})}\rho_i \proj{{\cB}_B({j})} \right )\right)\nonumber\\
        &&\quad-S (\rho_{AB})\nonumber\\
        &=& \min_{\mathcal{B}_B} \left(H(p) + \sum_{i=1}^{d_A} p_i S \left(\rho^{\cB_B}_{i} \right )\right) - \left(H(p) + \sum_{i=1}^{d_A} p_i S (\rho_i)\right)\nonumber\\
        &=& \min_{\mathcal{B}_B}\sum_{i=1}^{d_A} p_i \left[S \left(\rho^{\cB_B}_{i} \right ) -S (\rho_i)\right]\label{eq:CQbound1},
\end{eqnarray}
where $H(p)$ denotes the Shannon entropy of the probability distribution $\set{p_i}_i$, and the second equality follows from choosing $\mathcal{B}_A$ to coincide with the basis $\set{\ket{i}}_i$. Let $p_m:=\max_i p_i$. Our strategy is to let $\mathcal{B}_B$ project onto an eigenbasis of $\rho_m$, yielding:
\begin{eqnarray}
    Q (\rho_{AB})
        &\leq& \sum_{i\neq m} p_i \left[S \left(\rho^{\cB_B}_{i}\right ) -S (\rho_i)\right]\label{eq:simplestrat1}\\
        &\leq& \sum_{i\neq m} p_i S \left(\rho^{\cB_B}_{i} \right )\label{eq:simplestrat2} \\
        &\leq& \left(1-\frac{1}{d_A}\right)\log_2 d_B,\nonumber
\end{eqnarray}
where the second last inequality follows since $S(\rho_i)\geq 0$, and the last inequality follows since $p_m \geq 1/d_A$ and $S(\sigma_B)\leq \log_2 d_B$ for any density operator $\sigma_B$.
\end{proof}

For a two-qubit CQ state $\rho_{AB}$, Lemma~\ref{l:cqbound} implies $Q(\rho_{AB}) \leq 1/2$. We now show that this bound is tight by characterizing the set of CQ states attaining $Q (\rho_{AB}) = 1/2$.

\begin{lemma}\label{l:maxCQ}
    Consider CQ state $\rho_{AB}$ acting on $\complex^2\otimes\complex^2$ such that $\rho_{AB} = \sum_{i=1}^{2} p_i \ketbra{i}{i}\otimes\rho_i$, for some orthonormal basis $\set{\ket{i}}_{i=1}^{2}$. Then $Q (\rho_{AB}) = 1/2$ if and only if $p_1=p_2=1/2$ and $\rho_1=\proj{\psi_1}$ and $\rho_2=\proj{\psi_2}$ for some $\ket{\psi_1},\ket{\psi_2}\in\complex^2$ such that $\abs{\braket{\psi_1}{\psi_2}}^2=1/2$.
\end{lemma}
\begin{proof}
That $\rho_{AB}$ with $p_1\neq 1/2$ implies $Q(\rho_{AB}) < 1/2$ follows immediately from Eq.~\eqref{eq:simplestrat2} and the fact that $0 \leq S(\sigma)\leq 1$ for any 1-qubit density operator $\sigma$. We thus henceforth assume $p_1=p_2=1/2$. That $\rho_1$ and $\rho_2$ must be pure now also follows analogously, for if, say, $\rho_1$ is mixed, then we simply choose $\mathcal{B}_B$ in Eq.~\eqref{eq:simplestrat1} to instead project onto an eigenbasis of $\rho_2$, and use the fact that $S(\rho_1)>0$ to achieve $Q(\rho_{AB}) < 1/2$. We thus henceforth assume $\rho_1=\proj{\psi_1}$ and $\rho_2=\proj{\psi_2}$ for some $\ket{\psi_1},\ket{\psi_2}\in\complex^2$. It remains to show that we must have $\abs{\braket{\psi_1}{\psi_2}}^2=1/2$.

Plugging $\rho_{AB}$ into Eq.~\eqref{eq:CQbound1} and noting that $S(\rho_1)=S(\rho_2)=0$, we have
\begin{align}
    Q (\rho_{AB}) &=\frac{1}{2}\min_{\mathcal{B}_B} \left[S([\ketbra{\psi_1}{\psi_1}]^{\cB_B}) + S([\ketbra{\psi_2}{\psi_2}]^{\cB_B})\right]\\
    \begin{split}
    &= \frac{1}{2}\min_{\mathcal{B}_B} \Big[H\left(\abs{\braket{{\cB}_B({0})}{\psi_1}}^2,\abs{\braket{{\cB}_B({1})}{\psi_1}}^2\right)\\
    &\qquad\;\, \quad+H\left(\abs{\braket{{\cB}_B({0})}{\psi_2}}^2,\abs{\braket{{\cB}_B({1})}{\psi_2}}^2\right)\Big]\\
    \end{split}\\
    \begin{split}
    &= \frac{1}{2}\min_{\ket{{\cB}_B({0})}} \Big[H\left(\abs{\braket{{\cB}_B({0})}{\psi_1}}^2,\abs{\braket{{\cB}_B({0})}{\psi_1^\perp}}^2\right)\\
    &\;\,\quad\qquad \quad+ H\left(\abs{\braket{{\cB}_B({0})}{\psi_2}}^2,\abs{\braket{{\cB}_B({0})}{\psi_2^\perp}}^2\right)\Big]\label{eq:CQmin}
    \end{split}
\end{align}
where $\braket{\psi_1}{\psi_1^\perp}=\braket{\psi_2}{\psi_2^\perp}=0$, and where the last equality follows since $\proj{{\cB}_B({j})}$ are rank-one projectors. Note that one can think of the last equality as effectively switching the roles of the measurement and the target state, so that the minimization can be thought of as being taken over all pure \emph{target} states $\ket{{\cB}_B({0})}$ with respect to \emph{measurements} in the bases $\mathcal{B}_1:=\set{\ket{\psi_1},\ket{\psi_1^\perp}}$ and $\mathcal{B}_2:=\set{\ket{\psi_2},\ket{\psi_2^\perp}}$. We can now plug Eq.~\eqref{eq:CQmin} into the well-known entropic uncertainty relation of Maassen and Uffink\cite{Maassen1988} to immediately obtain:
\begin{equation}
    Q(\rho_{AB}) \geq \max _{\ket{\phi_1}\in \mathcal{B}_1,\ket{\phi_2}\in \mathcal{B}_2}-\log_2\abs{\braket{\phi_1}{\phi_2}}.
\end{equation}
Note that this lower bound attains its maximum value of $1/2$ if $\mathcal{B}_1$ and $\mathcal{B}_2$ are mutually unbiased, i.e. when $\abs{\braket{\psi_1}{\psi_2}}^2=1/2$. On the other hand, suppose $\mathcal{B}_1$ and $\mathcal{B}_2$ are \emph{not} mutually unbiased, i.e. suppose without loss of generality that $\abs{\braket{\psi_1}{\psi_2}}^2 > 1/2$. Then choosing $\ket{{\cB}_B({0})}=\ket{\psi_1}$ in Eq.~\eqref{eq:CQmin} yields $Q_{E_\textup{D}} (\rho_{AB})<1/2$. The claim follows.
\end{proof}

Combining Lemmas~\ref{l:cqbound} and~\ref{l:maxCQ}, we obtain a characterization of the set of two-qubit CQ states which are deemed maximally non-classical by $Q$. Such states include, for example, the CQ state
\begin{eqnarray}
    \rho &=& \frac{1}{2}\ketbra{0}{0}\otimes\ketbra{0}{0} + \frac{1}{2}\ketbra{1}{1}\otimes\ketbra{+}{+}\label{eq:CQexample}\\
    &=&\frac{1}{2}\left(
         \begin{array}{cccc}
           1 & 0 & 0 & 0 \\
           0 & 0 & 0 & 0 \\
           0 & 0 & \frac{1}{2} & \frac{1}{2} \\
           0 & 0 & \frac{1}{2} & \frac{1}{2} \\
         \end{array}
       \right),
\end{eqnarray}
where $\ket{+}=(\ket{0}+\ket{1})/\sqrt{2}$.

\section{Swapping the Ancilla-System Entanglement onto the System}
\label{sec:swapping}

We now explore the possibility of generating entanglement in the \emph{original} system ${\wek A}$ by projecting the ancilla systems $\wek{A}'$ of the state $\tilde{\rho}_{\wek{A}:\wek{A}'}$ of \eqref{eq:finalstate} jointly onto an entangled pure state. In other words, we consider a stochastic entanglement swapping process\cite{swapping} that maps the $\wek{A}:\wek{A}'$ entanglement onto the systems $\wek{A}$. As we are only interested in knowing whether this is possible (rather than, say, in the probability of success), the filtering via a pure state is not restrictive and corresponds to the best possible strategy. Our results indicate that this feat is possible for some, but not all, separable non-classical states.

We begin by noting that thanks to the maximally-correlated form \eqref{eq:maxcorr} of $\tilde{\rho}_{\wek{A}:\wek{A}'}$, we have that the (unnormalized) final state of system $\wek{A}$ after projecting the ancilla system onto the joint state $\ket{\phi}=\sum_\wek{i}\Phi_{\wek{i}}\ket{\wek{i}}$ for unit vector $\ket{\phi}\in(\complex^d)^{\otimes n}$ is given by
\begin{equation}
    \tilde{\rho}_{\wek{A}}
    =
    \Tr_{\wek{A'}}(\tilde\rho_{\wek{A}:\wek{A'}}\proj{\phi}_\wek{A'})
    =
    \sum_{\wek{i}\wek{j}}   \left[\rho_{\wek{ij}}^{\wek{\cB}}\Phi_{\wek{i}}\Phi^*_{\wek{j}}\right]\ketbra{\wek{i}}{\wek{j}}\label{eqn:state},
\end{equation}
with $\rho_{\wek{ij}}^{\wek{\cB}}$ defined in \eqref{eq:matrixelement}.
Hence, the resulting (unnormalized) state $\tilde{\rho}_{\wek{A}}$ is simply the Hadamard product of the original state (represented in the $\wek{\cB}$ basis) and $\ketbra{\phi}{\phi}$ (represented in the computational basis). Since $\ket{\phi}$ is arbitrary, we can say that $\tilde{\rho}_{\wek{A}}$ is obtained by rescaling rows and columns (with the same---up to conjugation---rescaling factor for row and column $\wek{i}$) of the original state $\rho_\wek{A}$ represented in the basis $\wek{\cB}$.

As previously mentioned, our goal is to answer the question of whether entanglement can be generated in the input systems for \emph{any} choice of starting local bases for the CNOT gates. In other words, while we allow arbitrary rescaling of rows and columns, the starting local basis $\wek{\cB}$ in which $\rho_\wek{A}$ is represented can be thought of as being chosen adversarially.

In Section~\ref{scn:iso} we provide a simple sufficient condition under which the generation of entanglement is always possible with an appropriate choice of $\ket{\phi}$, regardless of the choice of adversarial local bases $\wek{\cB}$. We then observe that this condition holds for all pseudo-isotropic states as in Eq.~\eqref{eq:pseudoisotropic}, with $\psi$ entangled and $p>0$.
In Sections~\ref{scn:CQ} and \ref{scn:QQ}, we  provide examples of Classical-Quantum (CQ) and Quantum-Quantum (QQ) separable states, respectively, for which entanglement between the systems cannot be generated in this fashion, i.e., there exists a choice of local unitaries that prevents the generation of entanglement in the systems via the swapping of system-ancilla entanglement, \emph{even if} there is necessarily entanglement between systems and ancillae after the activation protocol is run.

\subsection{Sufficient condition for entanglement swapping}
\label{scn:iso}


We focus again on the bipartite case $A_1=A$, $A_2=B$. We have the following simple condition which ensures the swapping of entanglement is possible.
\begin{theorem}
\label{thm:suffswap}
If for any choice of local basis $\wek{\cB}$, there exists a non-zero off-diagonal element of an off-diagonal block of $\rho_{AB}^{\wek{\cB}}$, i.e., if for all $\wek{\cB}=\cB_A\cB_B$ there exists a choice of $i\neq j$ and $k\neq l$ such that $\bra{\cB_A(i)\cB_B(k)}\rho_{AB}\ket{\cB_A(j)\cB_B(l)}\neq0$, then it is possible to swap entanglement back into the input systems (regardless of the choice of $\wek{\cB}$), i.e. there exists a $\ket{\phi}$ such that $\tilde{\rho}_{\wek{A}}$ in Eq.~(\ref{eqn:state}) is entangled.
\end{theorem}
\begin{proof}
   The strategy of the proof is to choose $\ket{\phi}$ so that the result of the Hadamard product in Eq.~(\ref{eqn:state}) is non-positive under partial transposition (NPT).\cite{peresPT,horodeckiPT} Fix any choice of local basis $\wek{\cB}$. By assumption, we know there exist indices $i\neq j$ and $k\neq l$ such that $\bra{\cB_A(i)\cB_B(k)}\rho_{AB}\ket{\cB_A(j)\cB_B(l)}\neq0$. In order to ensure that $\tilde{\rho}_A$ is NPT, we thus choose $\ket{\phi}$ to single out these non-zero off-diagonal terms by setting
       \begin{equation}
        \ket{\phi} = \frac{1}{\sqrt{2}}(\ket{ik} + \ket{jl}).
    \end{equation}
    With this choice of $\ket{\phi}$, $\tilde{\rho}_A$ becomes a Hermitian matrix with only four non-zero entries, two of which lie on the diagonal at positions $\ketbra{i}{i}\otimes\ketbra{k}{k}$ and $\ketbra{j}{j}\otimes\ketbra{l}{l}$, and two of which lie at off-diagonal positions of off-diagonal blocks at $\ketbra{i}{j}\otimes\ketbra{k}{l}$ and $\ketbra{j}{i}\otimes\ketbra{l}{k}$ (i.e. the four entries form the four corners of a square). It follows that the partial transpose of $\tilde{\rho}_A$ is not positive.
\end{proof}

\begin{corollary}\label{thm:iso}
    For any
    \beq
\label{eq:pseudoisotropic}
\rho(\psi,p)_{AB}:=(1-p)\frac{\openone_{AB}}{D}+p\proj{\psi}_{AB},
\eeq
with $\openone_{AB}/D$ the maximally mixed state for $AB$ and $D$ the dimension of $AB$, if $\ket{\psi}$ is entangled and $p>0$, then there exists a choice of $\ket{\phi}$ such that $\tilde{\rho}_{\wek{A}}= \Tr_{\wek{A'}}(\rho_{\wek{A}:\wek{A'}}\proj{\phi}_\wek{A'})$ is entangled.
\end{corollary}
\begin{proof}
Since the maximally mixed component of (18) is diagonal with respect to any choice of local bases, it suffices to argue that $|\psi\rangle$ satisfies the condition of Theorem 1. This easily follows from the fact $|\psi\rangle$ is entangled, and thus has, up to local unitaries, a Schmidt decomposition $\sum_{k=0}^{d_A-1} \sqrt{\lambda_k} \ket{k}\ket{k}$, with $\lambda_0\geq\lambda_1>0$.
\end{proof}

Corollary~\ref{thm:iso} shows that for any value of $p>0$, entanglement can be transferred to the original system for the pseudo-isotropic state $\rho(p,\psi)$ of Eq.~\eqref{eq:pseudoisotropic}, even for values of $p$ which correspond to \emph{separable} states (note that for $p$ small enough, the state $\rho(p,\psi)$ is separable due to the existence of a separable ball around the maximally mixed state\cite{zyczkowski1998sepball,gurvits2002sepball}). We remark that for all $p>0$ and entangled $\ket{\psi}$, $\rho(p,\psi)$ is non-classical,\cite{groismanquantumness,activation} and so here the non-classicality of the starting state allows us to create entanglement in the original systems $AB$ by applying the activation protocol followed by our entanglement swapping procedure.

\subsection{Classical-quantum separable states}\label{scn:CQ}

In Section~\ref{scn:iso}, we demonstrated that for certain non-classically correlated states, entanglement can be mapped back into the original system after the activation protocol is run. Can this be achieved with \emph{any} type of non-classically correlated input? We now show that the answer is no --- there exist separable non-classical states such that, while entanglement is always generated in the activation protocol between systems and ancillae independently of the local unitaries $U_A$ and $U_B$, a proper adversarial choice of local unitaries $U_A$ and $U_B$ can nevertheless prevent entanglement from being mapped back to the system.

Consider the separable non-classical CQ state of Eq.~\eqref{eq:CQexample}.
%
%
By Eq.~(\ref{eqn:state}), note that when the adversarial local unitaries are chosen as $U_A=U_B=\openone$, we have
\begin{equation}
    \tilde{\rho}_\wek{A}= \rho_{\wek{A}} \circ \ketbra{\phi}{\phi}.
\end{equation}
Since $\rho$ is block diagonal, it hence follows that $\tilde{\rho}_\wek{A}$ is block diagonal, since the Hadamard product cannot change this block diagonal structure regardless of the choice of $\ket{\phi}$. We conclude that there exists a choice of local bases (i.e the computational basis) with respect to which $\tilde{\rho}_\wek{A}$ is always separable for all $\ket{\phi}$, i.e., it is not possible to project the (necessarily present) system-ancilae entanglement generated in the activation protocol back onto the system. In fact, this proof approach holds for \emph{any} CQ (or QC) state that is not strictly classically correlated, implying that for such states, there is a choice of local unitaries for which, even if entanglement is created between system and ancilla in the activation protocol, such entanglement can not be swapped back into the input system.

\subsection{Quantum-quantum separable states}\label{scn:QQ}
Based on the results in Section~\ref{scn:CQ}, one might hope that entanglement generation in separable starting systems is possible if $\rho$ is not CQ nor QC (i.e. $\rho$ is what we might call \emph{QQ separable}). We provide a counterexample to this conjecture here --- namely, we prove that there exist QQ separable states for which an adversarial choice of local bases in the activation protocol prevents the swapping of ancilla-system entanglement back into the input systems.

To do so, consider the separable QQ operator:
\[
    \rho_{AB} = \frac{1}{2}\ketbra{0}{0}\otimes\ketbra{+}{+} + \frac{1}{2}\ketbra{+}{+}\otimes\ketbra{0}{0}
    =\frac{1}{4}\left(
                                                   \begin{array}{cccc}
                                                     2 & 1 & 1 & 0 \\
                                                     1 & 1 & 0 & 0 \\
                                                     1 & 0 & 1 & 0 \\
                                                     0 & 0 & 0 & 0 \\
                                                   \end{array}
                                                 \right).
\]
\noindent To prove our claim, as in Section~\ref{scn:CQ}, we choose local adversarial unitaries $U_A=U_B=\openone$ and show that $\tilde{\rho}_{AB}= \rho_{AB} \circ \ketbra{\phi}{\phi}$ is separable for any choice of $\ket{\phi}$. The latter is shown by first deriving a condition under which the eigenvalues of Hermitian operators with a structure similar to $\rho$ remain invariant under partial transposition. We then show that $\rho$ fulfills this condition for any choice of $\ket{\phi}$, implying $\rho$ always remains separable, since the partial transpose is a necessary and sufficient condition for separability of two-qubit states.\cite{horodeckiPT}

\begin{lemma}\label{l:invarSpec}
    Given any Hermitian operator $X$ acting on $\complex^2\otimes\complex^2$ with off-diagonal blocks which are diagonal, i.e.,
\begin{equation}
    X = \left(
      \begin{array}{cccc}
        a_{11} & a_{12} & a_{13} & 0 \\
        a_{12}^* & a_{22} & 0 & a_{24} \\
        a_{13}^* & 0 & a_{33} & a_{34} \\
        0 & a_{24}^* & a_{34}^* & a_{44} \\
      \end{array}
    \right),
\end{equation}
    if either $a_{12}a_{34}^\ast\in\reals$ or $a_{13}a_{24}^\ast\in\reals$, then the spectrum of $A$ is invariant under partial transposition.
\end{lemma}
\begin{proof}
    Let $p_X(\lambda)$ and $p_{X^\Gamma}(\lambda)$ denote the characteristic polynomials of $X$ and $X^\Gamma$, the partial transpose of $X$, respectively. Then $p_X(\lambda)-p_{X^\Gamma}(\lambda)=2\operatorname{Re}(a_{12}a_{34}^\ast a_{13} a_{24}^\ast-a_{12}^\ast a_{34} a_{13} a_{24}^\ast)=4 \operatorname{Im}( a_{13}^\ast a_{24} )\operatorname{Im}( a_{12} a_{34}^\ast)$, where $\operatorname{Re}(x)$ ($\operatorname{Im}(x))$ denotes the real (imaginary) part of $x$. The claim follows for $a_{12}a_{34}^\ast\in\reals$. An analogous calculation yields the $a_{13}a_{24}^\ast\in\reals$ case.
\end{proof}

With Lemma~\ref{l:invarSpec} in hand, it is easy to see that $\tilde{\rho}_{AB}$ has a positive partial transpose (and is hence separable) for all $\ket{\phi}$ --- specifically, we observe that $\rho$ satisfies the conditions of Lemma~\ref{l:invarSpec} since $a_{12}a_{34}^\ast=(1/4)(0)=0$, and this in particular holds even after taking the Hadamard product with any $\ketbra{\phi}{\phi}$. Since $\rho$ is positive semidefinite, it thus follows from Lemma~\ref{l:invarSpec} that $\tilde{\rho}_A$ must also be positive semidefinite under partial transposition and hence separable. Thus, there exist QQ separable states for which system-ancilla entanglement cannot be mapped back to the system.



Theorem \ref{thm:suffswap} tells us that if a two-qubit state $\rho$ has off-diagonal terms on its off-diagonal blocks for any choice of local bases, then entanglement can be created among the systems via swapping. On the other hand, if $\rho$ is restricted to having off-diagonal blocks which are diagonal, as was seen with the CQ and QQ counterexamples considered in Sections~\ref{scn:CQ} and~\ref{scn:QQ}, then there are choices of local initial rotations such that entanglement generation among the systems is not necessarily possible (actually, in the CQ case, entanglement generation is not possible for \emph{any} choice of local initial rotations).

One could ask whether this ``diagonal off-diagonal'' block structure is sufficient to rule out the possibility of entanglement generation. The answer is negative. Consider the following (un-normalized) positive semidefinite operator which has diagonal off-diagonal blocks:
\begin{equation}
    \rho=\left(
      \begin{array}{cccc}
        \frac{3}{2} & i & 1 & 0 \\
        -i & \frac{3}{2} & 0 & i \\
        1 & 0 & \frac{3}{2} & 1 \\
        0 & -i & 1 & \frac{3}{2} \\
      \end{array}
    \right).
\end{equation}
It turns out that the partial transposition of $\rho$ has a negative eigenvalue (observe that $\rho$ thus also necessarily violates the conditions of Lemma~\ref{l:invarSpec}). Hence, despite the fact that $\rho$ has off-diagonal blocks which are diagonal, it is nevertheless entangled, implying entanglement transfer to the system is possible for any choice of local bases: indeed, the Hadamard product can be chosen to be trivial, so that the projection simply gives back (a locally rotated and unnormalized) $\rho_\wek{A}$.

\section{Conclusions}
\label{sec:conclusions}

In this paper, stimulated by the findings of Ref.~\refcite{activation}, we have considered two issues: the quantification and bounding of non-classicality for classical-quantum states, and the interpretation of the activation protocol of Ref.~\refcite{activation} as a way to entangle input systems through interaction with ancillae. With respect to these two issues, we believe the most interesting open questions are the following.

We have found bounds on the non-classicality (as measured by the relative entropy of entanglement) of classical-quantum states, and we have characterized the maximally non-classical two-qubit classical-quantum states. It would be interesting to find bounds on the non-classicality of general separable states: from Ref.~\refcite{activation} we know that, e.g., a separable state of two qubits can never be as non-classical as a maximally entangled pure state, but at present we do not know how large the gap between the two is. It would also be nice to characterize the maximally non-classical classical-quantum states of systems of higher dimension than qubits.

With respect to the swapping of the post-activation ancilla-system entanglement onto the original systems, we have both necessary conditions and sufficient conditions for the swapping to be possible in an adversarial scenario, but we lack conditions which are \emph{simultaneously} necessary and sufficient. In finding such conditions, we suspect it would be beneficial to study the problem which arises in our swapping scheme: when is it possible to make a state entangled by rescaling rows and columns as in Eq.~\eqref{eqn:state}?

\section*{Acknowledgments}

We acknowledge  support by NSERC, QuantumWorks, CIFAR, Ontario Centres of Excellence, the Spanish MICINN through the Ram\'on y Cajal program, contract FIS2008-01236/FIS, and project QOIT (CONSOLIDER2006-00019), and by the Generalitat de Catalunya through
CIRIT 2009SGR-0985.

\bibliographystyle{ws-ijqi}
\bibliography{quantumness-ijqi}

\end{document}